\documentclass[11pt]{article}
\usepackage{cmap} %
\usepackage[paper=a4paper,left=1.91cm,right=1.91cm, top=2.54cm, bottom=2.54cm]{geometry}
\usepackage{amsmath, amsfonts, amssymb, amsthm, array, graphicx, float, morefloats, setspace, wrapfig}
\usepackage{icomma}
\usepackage{kpfonts}
\usepackage[T1]{fontenc}
\usepackage[utf8]{inputenc}
\usepackage[english]{babel}
\usepackage[font=small]{caption}

\usepackage[protrusion=true,expansion=all]{microtype}
\LoadMicrotypeFile{cmr}
\SetProtrusion
{ encoding = *,
 family = * }
{
« = {1000,     },
» = {    , 1000},
„ = {1000,     },
“ = {    , 1000},
” = {    , 1000},
( = {1000,     },
) = {    , 1000},
! = {    , 1000},
? = {    , 1000},
: = {    , 1000},
; = {    , 1000},
. = {    , 1000},
- = {    ,  500},
— = {    ,  500},
{,} = {    , 1000}
}
\theoremstyle{definition}

\usepackage[title]{appendix}
\usepackage{appendix}

\usepackage{lscape}
\usepackage{rotating}
\usepackage{booktabs}
\usepackage{tablefootnote}
\usepackage{threeparttable}

\usepackage{color}

\usepackage[section]{placeins}
\usepackage{float}

\usepackage[colorlinks=true]{hyperref}
\hypersetup{colorlinks, citecolor=cyan, linkcolor=magenta, urlcolor=blue}

\newcommand{\argmin}{\operatornamewithlimits{argmin}}

\usepackage{tikz}

\usepackage{verbatim}

\usepackage{subcaption}
\usepackage{mwe}

\usepackage[font={small,bf},labelfont=bf,skip=0pt]{caption}
\usepackage{bm} %
\usepackage{tcolorbox} %

\usepackage{dsfont}
 
\newtheorem{theorem}{Theorem}
\newtheorem{exmp}{Example}
\newtheorem{lemma}{Lemma}
\newtheorem{assumption}{Assumption}

\usepackage[ruled]{algorithm2e}
\newcommand{\Break}{\textbf{break}}

\renewenvironment{abstract}
{\centerline{\large\bf Abstract}\vspace{0.7ex}%
  \bgroup\leftskip 20pt\rightskip 20pt\small\noindent\ignorespaces}%
{\par\egroup\vskip 0.25ex}

\newenvironment{keywords}
{\vspace{1em}\bgroup\leftskip 20pt\rightskip 20pt \small\noindent{\bf Keywords:} }%
{\par\egroup\vskip 0.25ex}

\usepackage[backend=biber, bibencoding=utf8, sorting=nty, style=authoryear-icomp,
maxcitenames=2, maxbibnames=99, style=bwl-FU, uniquename=false, uniquelist=false, natbib=true]{biblatex}
\usepackage[affil-it]{authblk}

\title{Causal Gradient Boosting: \\ Boosted Instrumental Variable Regression}

\author[1]{Edvard Bakhitov}
\author[2]{Amandeep Singh}
\affil[1]{Department of Economics, University of Pennsylvania. E-mail: bakhitov@sas.upenn.edu}
\affil[2]{The Wharton School, University of Pennsylvania. E-mail: amansin@wharton.upenn.edu}

\date{\today}

\begin{document}

\maketitle

\begin{abstract}
Recent advances in the literature have demonstrated that standard supervised learning algorithms are ill-suited for problems with endogenous explanatory variables. To correct for the endogeneity bias, many variants of nonparameteric instrumental variable regression methods have been developed. In this paper, we propose an alternative algorithm called \textit{boostIV} that builds on the traditional gradient boosting algorithm and corrects for the endogeneity bias. The algorithm is very intuitive and resembles an iterative version of the standard 2SLS estimator. Moreover, our approach is data driven, meaning that the researcher does not have to make a stance on neither the form of the target function approximation nor the choice of instruments. We demonstrate that our estimator is consistent under mild conditions.  We carry out extensive Monte Carlo simulations to demonstrate the finite sample performance of our algorithm compared to other recently developed methods. We show that boostIV is at worst on par with the existing methods and on average significantly outperforms them. 
\end{abstract}


\section{Introduction}

Gradient boosting method is considered one of the leading  machine learning (ML) algorithms for supervised learning with structured data. There is a large body of evidence showing that gradient boosting dominates in a significant number of ML competitions conducted on Kaggle\footnote{For reference see \url{https://www.kaggle.com/dansbecker/xgboost}}. However, recent literature \citep[e.g., see][]{hartford2017} has shown that traditional supervised machine learning methods do not perform well in the presence of endogeneity in the explanatory variables. 

A common approach to correct for the endogeneity bias is to use instrumental variables (IVs). Nonparametric instrumental variables (NPIV) methods have regained popularity among applied researchers over the last decade as they do not require imposing (possibly) implausible parametric assumptions on the target function. On the other hand, existing nonparameteric estimation techniques require the researcher to specify a target function approximation (ideally driven by some ex-ante understanding of the data generating process), e.g. a sieve space, which in turn drives the choice of unconditional moment restrictions (or simply put, the choice of IV basis functions). If the approximation is bad, it will lead to misspecification issues, and if the IVs are ``weak''\footnote{By ``weak'' we mean that the IV basis functions are weakly correlated with the target function basis functions. It is hard to define a notion of weak instruments in the NPIV set-up of \cite{newey_powell2003}, since there is no explicit reduced form. In the triangular simultaneous equations models, where the explicit reduced form exists, \cite{han2014} defined weak IVs as a sequence of reduced-form functions where the associated rank shrinks to zero.}, most likely the standard NPIV asymptotic techniques will no longer be valid. Moreover, the complexity of both modelling and estimation explodes when there are more than a handful of inputs.

In this paper, we introduce an algorithm that allows to learn the target function in the presence of endogenous explanatory variables in a data driven way, meaning that the researcher does not have to make a stance on neither the form of the target function approximation nor the choice of instruments. We build on gradient boosting algorithm to transform the standard NPIV problem into a learning problem that accounts for endogeneity in explanatory variable, and thus, we call our algorithm \textit{boostIV}.

We also consider a couple extensions to the boostIV algorithm that might improve its finite sample performance. First, we show how to incorporate optimal IVs, i.e. IVs that achieve the lowest asymptotic variance \citep{chamberlain1987asymptotic}. Second, we augment the boostIV algorithm with a post-processing step where we re-estimate the weights on the learnt basis functions, we call this algorithm \textit{post-boostIV}. The idea is based on \cite{friedman_popescu2003} who propose to learn an ensemble of basis functions and then apply lasso to perform basis function selection.

To avoid potentially severe finite sample bias due to the double use of data, we resort to the cross-fitting idea of \cite{chernozhukov2017dml}. For the boostIV algorithm we split the data to learn instruments and basis functions on different data folds. We add an additional layer of cross-fitting to the post-boostIV algorithm to update the weights on the learnt basis functions.

Our method has a number of advantages over the standard NPIV approach. First, our approach allows the researcher to be completely agnostic to the choice of basis functions and IVs. Both basis functions and instruments are learnt in a data driven way which picks up the underlying data structure. Second, the method becomes even more attractive when the dimensionality of the problem grows, as the standard NPIV methods suffer greatly from the curse of dimensionality. Intuitively, learning via boosting should be able to construct basis functions that approximately represent the underlying low dimensional data features. However, our approach does not work in purely high-dimensional settings where the number of regressors exceeds the number of observations. 

We study the performance of boostIV and post-boostIV algorithms in a series of Monte Carlo experiments. We compare the performance of our algorithms to both the standard sieve NPIV estimator and a variety of modern ML estimators.  Our results demonstrate that boostIV performs at worst on par with the state of the art ML estimators. Moreover, we find no empirical evidence that post-boostIV achieves superior performance compared to boostIV and vice versa. However, adding the post-processing step reduces the amount of boosting iterations needed for the algorithm to converge rendering it (potentially) computationally more efficient\footnote{To be more precise, there is a trade-off at play. One boostIV iteration takes less time than one post-boostIV iteration as the latter algorithm includes an additional estimation step plus one more layer of cross-fitting. As a result, if adding the post-processing step reduces the amount of boosting iterations significantly, then we achieve computational gains. It might not be the case otherwise.}. 

This paper brings together two strands of literature. First, our approach contributes to the literature on nonparametric instrumental variables modeling. \cite{newey_powell2003} propose to replace the linear relationships in standard linear IV regression with linear projections on a series of basis functions (also see \cite{blundell2007_SNPIV} for an application to Engel-curve estimation).  \cite{darolles2011nonparametric} and \cite{hall2005nonparametric} suggest to nonparametrically estimate the conditional distribution of endogenous regressors given the instruments, $F(x|z)$, using kernel density estimators. However, despite their simplicity and flexibility, both approaches are subject to the curse of dimensionality. Machine learning literature has recently also contributed to the nonparametric IV literature. \cite{hartford2017} propose a DeepIV estimator which first estimates $F(x|z)$ with a mixture of deep generative models on which then the structural function is learned with another deep neural network. Kernel IV estimator of \cite{singh2019kiv} exploits conditional mean embedding of $F(x|z)$, which is then used in the second stage kernel ridge regression. \cite{muandet2019} avoid the traditional two stage procedure by focusing on the dual problem and fitting just a single kernel ridge regression.

Second, we exploit insights from the boosting literature. Originally boosting came out as an ensemble method for classification in the computational learning theory \citep{schapire1990strength,freund1995boosting, freund1997AdaBoost}. Later on \cite{friedman2000additive} draw connections between boosting and statistical learning theory by viewing boosting as an approximation to additive modeling. A different perspective on boosting as a gradient descent algorithm in a function space that connects boosting to the more common optimization view of statistical inference \citep{breiman1998arcing,breiman1999prediction,friedman2001greedy}. $L_{2}$-boosting introduced by \cite{buhlmann2003L2boost} provides a powerful tool to learning regression functions. A comprehensive boosting review can be found in \cite{buhlmann2007boosting}.

The remainder of the paper is organized as follows. Section \ref{sec:setup} briefly introduces the NPIV framework. Section \ref{sec:boost_review} describes the standard boosting procedure. We present boostIV and post-boostIV in Section \ref{sec:boost_npiv}. Section \ref{sec:cv} talks about hyperparameter tuning. Section \ref{sec:consistency} discusses consistency. We illustrate the numerical performance of our algorithms in Section \ref{sec:mc}. Section \ref{sec:conclusion} concludes. All the proofs and mathematical details are left for the Appendix.

\section{Set-up} \label{sec:setup}

Consider the standard conditional mean model of \cite{newey_powell2003}
\begin{equation} \label{eq:model_eq}
	y = g(x) + \varepsilon, \quad \mathbb{E}[\varepsilon|z] = 0,
\end{equation}
where $y$ is a scalar random variable, $g$ is an unknown structural function of interest, $x$ is a $d_{x} \times 1$ vector of (potentially) endogenous explanatory variables, $z$ is a $d_{z} \times 1$ vector of instrumental variables, and $\varepsilon$ is an error term\footnote{The approach can easily be extended to cases where only some of the regressors are endogenous. Suppose $x = (x_{1},\,x_{2})$ where $x_{1}$ consists of endogenous regressors and $x_{2}$ is a vector of exogenous regressors. Let $w$ be a vector of excluded instruments and set $z = (w,\,x_{2})$. This perfectly fits into the model described by \eqref{eq:model_eq}.}. Suppose that the model is identified and the completion condition holds, i.e. for all measurable real functions $\delta$ with finite expectation,
\begin{equation*}
	\mathbb{E}[\delta(x)|z] = 0 \Rightarrow \delta(x) = 0.
\end{equation*}
Intuitively this condition implies that there is enough variation in the instruments to explain the variation in $x$.

The conditional expectation of \eqref{eq:model_eq} yields the integral equation
\begin{equation} \label{eq:cond_exp}
	\mathbb{E}[y|z] = \int g(x)dF(x|z),
\end{equation}
where $F$ denotes the conditional cdf of $x$ given $z$. Solving for $g$ directly is an ill-posed problem as it involves inverting linear compact operators (see e.g. \cite{kress1989}). Note that the model in \eqref{eq:model_eq} does not have an explicit reduced form, i.e. a functional relationship between endogenous and exogenous variables, however, it is implicitly embedded in $F$. Thus, from the estimation perspective we have two objects to estimate: (i) the conditional cdf F(x|z) and (ii) the structural function $g$.

A common approach in applied work is to assume that the relationships between $y$ and $x$ as well as $x$ and $z$ are linear, which leads to a standard 2SLS estimator. However, it can be a very restrictive assumption in practice, which can result in misspecification bias. A lot of more flexible non-parametric extension to 2SLS have been developed in the econometrics literature. The standard approach is ti use the series estimator of \cite{newey_powell2003} who propose to replace the linear relationships with a linear projections on a series of basis functions. 

To illustrate the approach let us approximate $g$ with a series expansion
\begin{equation*}
	g(x) \approx \sum_{\ell=1}^{L} \gamma_{\ell} p_{\ell}(x),
\end{equation*}
where $p^{L}(x) = (p_{1}(x),\dots,p_{L}(x))$ is a series of basis functions. It allows us to rewrite the conditional expectation of $y$ given $z$ as
\begin{equation} \label{eq:cond_exp}
	\mathbb{E}[y|z] \approx \sum_{\ell=1}^{L} \gamma_{\ell}\mathbb{E}[p_{\ell}(x)|z].
\end{equation}
Let $q^{K}(z) = (q_{1}(z),\dots,q_{K}(z))$ be a series of IV basis functions. This implies a 2SLS type estimator of $\gamma$
\begin{equation}
	\hat{\gamma} = \left(\hat{\mathbb{E}}[p^{L}(x)|z]'\hat{\mathbb{E}}[p^{L}(x)|z]\right)^{-}\hat{\mathbb{E}}[p^{L}(x)|z]'y,
\end{equation}
where $\hat{\mathbb{E}}[p^{L}(x)|z] = q^{K}(z)\left(q^{K}(z)'q^{K}(z)\right)^{-}q^{K}(z)'p^{L}(x)$. Given $L,K \rightarrow \infty$ as $n \rightarrow \infty$, asymptotically one can recover the true structural function. However, in finite samples one has to truncate the sieve at some value. Despite that, the performance of the estimator hinges crucially on the choice of the approximating space, especially in high dimensions. Moreover, NPIV estimators suffer greatly from the curse of dimensionality which renders them inapplicable even in many applications. Alternatively, we propose a data-driven approach, which is agnostic to the choice of sieve/approximating functions.

\section{Revisiting Gradient Boosting} \label{sec:boost_review}
 
Boosting is a greedy algorithm to learn additive basis function models of the form
\begin{equation} \label{eq:boost_def}
	f(x) = \alpha_{0} + \sum_{m=1}^{M}\alpha_{m}\varphi(x;\theta_{m}),
\end{equation}
where $\varphi_{m}$ are generated by a simple algorithm called a \textit{weak learner} or \textit{base learner}. The weak learner can be any classification or regression algorithm, such as a regression tree, a random forest, a simple single-layer neural network, etc. One could boost the performance (on the training set) of any weak learner arbitrarily high, provided the weak learner could always perform slightly better than chance\footnote{This is relevant when applied to classification problems. For regression problems any simple method such as least squares regression, regression stump, or one or two-layered neural network will work.} \citep{schapire1990strength, freund1996experiments}. It is a very nice feature, since the only thing we need to make a stance on is the form of the weak learner, which is much less restrictive than choosing a sieve. 

The goal of boosting is to solve the following optimization problem
\begin{equation} \label{eq:boosting_problem}
	\min_{f} \sum_{i=1}^{N} L(y_{i},f(x_{i})),
\end{equation}
where $L(y,y')$ is a loss function and $f$ is defined by \eqref{eq:boost_def}. Since the boosting estimator depends on the choice of the loss function, the algorithm to solve \eqref{eq:boosting_problem} should be adjusted for a particular choice. Instead, one can use a generic version called \textit{gradient boosting} \citep{friedman2001greedy, mason2000boosting}, which works for an arbitrary loss function. 

\cite{breiman1998arcing} showed that boosting can be interpreted as a form of the gradient descent algorithm in function space. This idea then was further extended by \cite{friedman2001greedy} who presented the following functional gradient descent or gradient boosting algorithm:

\begin{enumerate}
    \item Given data $\{(y_{i},\,x_{i})\}_{i=1}^{n}$, initialize the algorithm with some starting value. Common choices are
    \begin{equation*}
	    f_{0}(x) \equiv \argmin_{c} \sum_{i=1}^{N} L(y_{i},c),
    \end{equation*}
    which is simply $\bar{y}$ under the squared loss, or $f_{0}(x) \equiv 0$. Set $m = 0$.
    \item Increase $m$ by 1. Compute the negative gradient vector and evaluate it at $f_{m-1}(x_{i})$:
    \begin{equation*}
        r_{im} = - \left.\frac{\partial L(y_{i},\,f)}{\partial f}\right\vert_{f = f_{m-1}(x_{i})}, \quad i = 1,\dots,\,n.
    \end{equation*}
    \item Use the weak learner to compute $(\alpha_{m},\,\theta_{m})$ which minimize $\sum_{i=1}^{N} (r_{im} - \alpha \phi(x_{i};\theta))^{2}$.
    \item Update 
    \begin{equation*}
	    f_{m}(x) = f_{m-1}(x) + \alpha_{m} \phi(x;\theta_{m}),
    \end{equation*}
    that is, proceed along an estimate of the negative gradient vector. In practice, better (test set) performance can be obtained by performing ``partial updates'' of the form
    \begin{equation*}
	    f_{m}(x) = f_{m-1}(x) + \nu \alpha_{m} \phi(x;\theta_{m}),
    \end{equation*}
    where $0 \leq \nu \leq 1$ is a shrinkage parameter, usually set close to zero \citep{friedman2001greedy}. 
    \item Iterate steps 2 to 4 until $m = M$ for some stopping iteration $M$.
\end{enumerate}

The key point is that we do not go back and adjust earlier parameters. The resulting basis functions learnt from the data are $\phi(x) = (\phi(x;\theta_{1}),\dots,\phi(x;\theta_{M}))$. The number of iterations $M$ is a tuning parameter, which can be optimally tuned via cross-validation or some model selection criterion (see Section \ref{sec:cv} for more details).

\section{Boosting the IV regression} \label{sec:boost_npiv}

The main complication in the NPIV set-up is that $x$ is potentially endogenous, otherwise learning the structural function via boosting would be straightforward. Moreover, we cannot learn basis functions in the first step and then construct IVs in the second. Dependence of the basis functions for the structural equation on the instruments and vice versa suggests an iterative algorithm. 

Before we introduce the algorithm, we need to set up the boosting IV framework first. Combining \eqref{eq:model_eq} and  \eqref{eq:boost_def} gives
\begin{equation} \label{eq:boost_struct_eq}
	y = \alpha_{0} + \sum_{m=1}^{M}\alpha_{m}\varphi(x;\theta_{m}) + \varepsilon.
\end{equation}
Hence, the conditional expectation of $y$ given $z$ becomes
\begin{equation} \label{eq:boost_cond_exp}
	\mathbb{E}[y|z] = \alpha_{0} + \sum_{m=1}^{M}\alpha_{m}\mathbb{E}[\varphi(x;\theta_{m})|z].
\end{equation}
Note that \eqref{eq:boost_struct_eq} and \eqref{eq:boost_cond_exp} closely resemble their standard NPIV counterparts \eqref{eq:model_eq} and \eqref{eq:cond_exp}. The only difference is that the form of basis functions for boosting must be estimated, while for the standard NPIV it has to be ex-ante specified. Unlike standard boosting, where the goal is to learn $\mathbb{E}[y|x]$, in the presence of endogeneity, we want to learn $\mathbb{E}[y|x]$, implying that in each boosting iteration we have to learn the conditional expectation of the weak learner given the IVs.

To keep things clear and simple, we focus on $L_{2}$-boosting which assumes the squared loss function. \cite{buhlmann2003L2boost} show that $L_{2}$-boosting is equivalent to iterative fitting of residuals. In the IV context, it means that at step $m$ the loss has the form
\begin{equation*}
	L(y, g_{m-1}(x) + \alpha\mathbb{E}[\varphi(x;\theta)|z]) = (r_{m} - \alpha\mathbb{E}[\varphi(x;\theta)|z])^2, 
\end{equation*}
where $r_{m} \equiv y - g_{m-1}$ is the current residual. Thus, at step $m$ the optimal parameters minimize the loss between the residuals and the conditional expectation of the weak learner given the instruments, 
\begin{equation} \label{eq:boostIV_iteration}
	(\alpha_{m},\,\theta_{m}) = \argmin_{\alpha, \theta} \sum_{i=1}^{N}(r_{im} - \alpha\mathbb{E}[\varphi(x_{i};\theta)|z_{i}])^2.
\end{equation}
However, the conditional expectation $\mathbb{E}[\varphi(x;\theta)|z]$ is unknown and has to be estimated. 

A simple way to estimate the conditional expectation in \eqref{eq:boostIV_iteration} is to project\footnote{In general we do not have to use a projection, we can use a more complex model to estimate the conditional expectation.} the weak learner on the space spanned by IVs
\begin{equation*}
	\hat{\mathbb{E}}[\varphi(x;\theta)|z] = P_{Z}\varphi(x;\theta),  
\end{equation*}
where $P_{A} = A(A'A)^{-1}A'$ is a projection matrix. The exogeneity condition in \eqref{eq:model_eq} implies that any function of $z$ can serve as an instrument. However, we do not need any function, we need such a transformation of $z$ that will give us strong instruments, i.e. instruments that explain the majority of the variation in the endogenous variables. We follow \cite{gandhi2019cross_learner} and introduce an additional step on which we learn the instruments. Let $\mathcal{H}(\cdot;\eta)$ be a class of IV functions parameterized by $\eta$. This formulation allows us to use various off-the-shelf algorithms such as Neural Networks, Random Forests, etc. to learn $\mathcal{H}(\cdot;\eta)$. Given the learnt IV transformation $\mathcal{H}(\cdot;\eta)$, we can rewrite \eqref{eq:boostIV_iteration} as
\begin{equation*}
	(\alpha_{m},\,\theta_{m}) = \argmin_{\alpha, \theta} \sum_{i=1}^{N}(r_{im} - \alpha P_{\mathcal{H}(z_{i};\eta)}\varphi(x_{i};\theta))^2.
\end{equation*}

Since the basis function parameters $(\alpha,\,\theta)$ depend on the IV transformation parameters $\eta$ and vice verse, we propose an algorithm that iterates between two steps. At the first step we learn instruments, i.e. $\eta_{m}$, given the basis functions parameter estimates from the previous iteration $(\alpha_{m-1},\,\theta_{m-1})$, then at the second step we learn new parameter estimates $(\alpha_{m},\,\theta_{m})$ given the instruments from the first step. We can draw an analogy with the canonical two-stage least squares, where we estimate the reduced form in the first stage, and the structural equation in the second. The details are provided in Algorithm \ref{algo:naive}.

\begin{algorithm}[!h]
	\caption{Naive boostIV} 
		Initialize basis functions: $\varphi_{0} = \bar{y}$\\
			\For {iteration $m$} {
				{\textbf{First stage:} given $\varphi(x;\theta_{m-1})$, estimate $\mathcal{H}(z;\eta_{m})$ \\
				\textbf{Second stage:} given $\mathcal{H}(z;\eta_{m})$, solve
				\begin{equation*}
					(\alpha_{m},\,\theta_{m}) = \argmin_{\alpha, \theta} \sum_{i=1}^{N}\left(r_{im} - \alpha P_{\mathcal{H}(z_{i};\eta_{m})}\varphi(x_{i};\theta)\right)^2
				\end{equation*} \\
				\textbf{update:} $g_{m}(x) = g_{m-1}(x) + \alpha_{m}\varphi(x;\theta_{m})$}
				}
		 Stop at iteration $M$
	\label{algo:naive}
\end{algorithm}

We call this algorithm the Naive boostIV, since we use the same data to learn both the instruments and the basis functions. Asymptotically this will not affect the properties of the estimator, however, in finite samples biases from the first stage will propagate to the second. This issue can be especially severe if we use regularized estimators in the first stage as the regularization bias will heavily affect the second stage estimates. To get around this issue we resort to cross-fitting.

Let $D = \{y_{i},x_{i},z_{i}\}_{i = 1}^{n}$ be our data set, where $D_{i}$ are iid. Split the data set into a $K$-fold partition, such that each partition $D_{k}$ has size $\left\lfloor \frac{n}{K} \right\rfloor$, and let $D_{k}^{c}$ be the excluded data. The boostIV procedure with cross-fitting is described in Algorithm \ref{algo:boost_cf}.

\begin{algorithm}
	\caption{boostIV with cross-fitting} 
		Folds $\{\mathcal{D}_{1},\,\dots,\,\mathcal{D}_{K}\}\leftarrow$ \textsc{Partition}$(\mathcal{D},\,K)$ \\
	    Initialize basis functions: $\varphi_{0}^{k} = \bar{y}$ for $k=1,\,\dots,\,K$ \\
			\For {iteration $m$}{ 
			\For {fold $k$}{
				\textbf{First stage:} 
				\begin{itemize}
					\item given $\varphi(x_{k}^{c};\theta^{k}_{m-1})$ and $z_{k}^{c}$, estimate $\mathcal{H}(\cdot;\eta^{k}_{m})$
					\item apply the learnt transformation to generate  IVs $\mathcal{H}(z_{k};\eta^{k}_{m})$
				\end{itemize}
				\textbf{Second stage:} Given $\mathcal{H}(z_{k};\eta^{k}_{m})$, solve
				\begin{equation*}
					(\alpha^{k}_{m},\,\theta^{k}_{m}) = \argmin_{\alpha, \theta} \sum_{i \in \mathcal{D}_{k}} \left(r_{im} - \alpha P_{\mathcal{H}(z_{i};\eta^{k}_{m})}\varphi(x_{i};\theta)\right)^2
				\end{equation*} \\
				\textbf{update:} $g^{k}_{m}(x_{k}) = g^{k}_{m-1}(x_{k}) + \alpha^{k}_{m}\varphi(x_{k};\theta^{k}_{m})$
			}
			}
		 Stop at iteration M \\
		 \KwOut{$\hat{g}(x) = \frac{1}{K}\sum_{k=1}^{K}g_{M}^{k}(x)$}
	\label{algo:boost_cf}
\end{algorithm}

\subsection{Learning optimal instruments}

Our boostIV algorithm also allows to incorporate optimal instruments in the sense of \cite{chamberlain1987asymptotic}, i.e. instruments that achieve the smallest asymptotic variance.  Assuming conditional homoskedasticity, the optimal instrument vector of \cite{chamberlain1987asymptotic} at step $m$ is
\begin{equation} \label{eq:opt_chamb}
	\mathcal{H}(z;\eta_{m}) = D_{m}(z)\sigma_{m}^{-2},
\end{equation}
where 
\begin{equation} \label{eq:opt_deriv}
	D_{m}(z) = \mathbb{E}\left[\left.\frac{\partial \varepsilon(\gamma_{m})}{\partial \gamma_{m}'}\right\vert z\right], \quad \gamma_{m}=(\alpha_{m},\,\theta_{m}')'
\end{equation}
is the conditional expectation of the derivative of the conditional moment restriction with respect to the boosting parameters, and $\sigma^{2}_{m} = \mathbb{E}[r_{m}^{2}|z]$ is the conditional variance of the error term at step $m$. Thus, the IV transformation parameters $\eta_{m}$ are implicitly embedded in a particular approximation used to estimate $D_{m}(z)$. 

The main complication with using optimal IVs is that they are generally unknown, hence, the common approach is to consider approximations. The parametrization in \eqref{eq:opt_chamb}-\eqref{eq:opt_deriv} allows us to use any off-the-shelf statistical/ML method to estimate the optimal functional form for the instruments. Moreover, the iterative nature of the algorithm allows us to use the estimates from step $m-1$ as proxies. 

\subsection{Post-processing}

An important feature of the forward stage-wise additive modeling is that we do not go back and adjust earlier parameters. However, we might want to revisit the weights on the learnt basis functions to achieve a better fit. This can be seen as a way of post-processing our boostIV procedure. 

The whole procedure can be broken down into two stages:
\begin{enumerate}
	\item Apply the boostIV algorithm to learn basis functions $\hat{\varphi}_{m}(x) = \frac{1}{K}\sum_{k=1}^{K}\varphi(x;\theta_{m}^{k})$ for $m=1,\,\dots,\,M$;
	\item Estimate the weights
	\begin{equation} \label{eq:post_boost_no_cf}
		\hat{\beta} = \argmin_{\beta}\sum_{i=1}^{n}\left(y_{i} - \beta_{0} - \sum_{m=1}^{M}\beta_{m}\hat{\varphi}_{m}(x_{i})\right)^2.
	\end{equation}
\end{enumerate}
Note that the basis functions $(\hat{\varphi}_{1}(x),\,\dots,\,\hat{\varphi}_{M}(x))$ are causal in the sense that they are constructed using estimated parameters $\theta$ that identify a causal relationship between $x$ and $y$. 

Boosting is an example of an ensemble method which combines various predictions with appropriate weights to get a better prediction. In the context of boostIV it works in the following way. We exploit the variation in IVs to get causal parameters $\theta$. Given the estimated parameters, we can treat each learnt basis function $\varphi(x;\theta_{m})$ as a separate prediction obtained by fitting a base learner. Then, the post-processing step in \eqref{eq:post_boost_no_cf} can be simply seen as model averaging.

We can estimate optimal weights $\hat{\beta}$ by simply running a least squares regression as in \eqref{eq:post_boost_no_cf} or use any other method such as Random Forrests, Neural Networks, boosting, etc. To avoid carrying over any biases from the estimation of $(\hat{\varphi}_{1}(x),\,\dots,\,\hat{\varphi}_{M}(x))$ into the choice of $\beta$, we use cross-fitting once again, which is a generalization of the stacking idea of \cite{wolpert1992}. 

\begin{algorithm}[!h]
	\caption{post-boostIV} 
		Folds $\{\mathcal{D}_{1},\,\dots,\,\mathcal{D}_{L}\}\leftarrow$ \textsc{Partition}$(\mathcal{D},\,L)$ \\
			\For {fold $\ell$} {
			\begin{enumerate}
				\item[1.] apply boostIV to $\mathcal{D}_{\ell}^{c}$ and estimate basis functions $(\hat{\varphi}^{\ell}_{1}(x),\,\dots,\,\hat{\varphi}^{\ell}_{M}(x))$
				\item[2.] estimate post-boosting weights 
				\begin{equation*}
					\hat{\beta}^{\ell} = \argmin_{\beta}\sum_{i \in \mathcal{D}_{\ell}}\left(y_{i} - \beta_{0} - \sum_{m=1}^{M}\beta_{m}\hat{\varphi}^{\ell}_{m}(x_{i})\right)^2.
				\end{equation*}
				\item[3.] fold fit at point $x$: $g^{\ell}(x) = \hat{\beta}^{\ell}_{0} + \sum_{m=1}^{M}\hat{\beta}^{\ell}_{m}\hat{\varphi}^{\ell}_{m}(x)$
			\end{enumerate}	
			}
		 Stop at iteration $M$ \\
		 \KwOut{$\hat{g}(x) = \frac{1}{L}\sum_{\ell=1}^{L}g_{M}^{\ell}(x)$}
	\label{algo:post-boost}
\end{algorithm}

\section{Choosing the optimal number of boosting iterations} \label{sec:cv}

Boosting performance crucially depends on the number of boosting iterations, in other words, $M$ is a tuning parameter. A common way to tune any ML algorithm is cross-validation (CV). The most popular type of CV is $k$-fold CV. The idea behind $k$-fold CV is to create a number of partitions (validation datasets) from the training dataset and fit the model to the training dataset (sans the validation data). The model is then evaluated against each validation dataset and the results are averaged to obtain the cross-validation error. In application to boosting, we can estimate the CV error for a grid of candidate tuning parameters (number of iterations) and pick $M^{*}$ that minimizes the CV error. Alternatively, \cite{buhlmann2007boosting} show how to apply AIC and BIC criteria to boosting in the exogenous case. However, it is not clear how to adjust those criteria for the presence of endogeneity.

Both standard $k$-fold cross validation and the model selection criteria considered in \cite{buhlmann2007boosting} can be computationally costly as it is necessary to compute all boosting iterations under consideration for the training data. To surpass this issue, we apply \textit{early stopping} to $k$-fold CV. The idea behind early stopping is to monitor the behavior of the CV error and stop as soon as the performance starts decreasing, i.e. CV error goes up. 

Algorithm \ref{algo:CV_ES} provides implementation details for the $k$-fold CV with early stopping for either boostIV or post-boostIV procedure. The early stopping criterion compares the CV error evaluated for the model based on $M_j$ boosting iterations to the CV error evaluated for the model based on $M_i$, $M_i < M_j$. If $CV^{err}(M_j) > CV^{err}(M_i) + \epsilon$, where $\epsilon > 0$ but close to zero is a numerical error tolerance level, then we stop and set $M^{*} = M_i$, otherwise, continue the search. If the criterion is not met for any of the candidate tuning parameters, we pick the largest value $M^{*} = \bar{M}$.

An alternative solution would be to use a slice of the dataset as the validation sample and tune the number of iterations using the observations from the validation sample. We actually use this approach in our simulations since it significantly reduces the computational burden.
\begin{algorithm}[!h]
	\caption{$k$-fold CV with early stopping} 
		Folds $\{\mathcal{D}_{1},\,\dots,\,\mathcal{D}_{k}\}\leftarrow$ \textsc{Partition}$(\mathcal{D},\,k)$ \\
		Set of indices $\mathcal{I}_{M}$ corresponding to a sorted grid of tuning parameters $\mathcal{M} = \{1,\,\dots,\,\bar{M}\}$ \\
			\While {$\mathcal{M}[i] \leq \bar{M}$ for $i \in \mathcal{I}_{M}$} {
			\For {fold $\kappa=1,\,\dots,\,k$} {
				\begin{enumerate}
					\item[1.] training set $\mathcal{T}_{\kappa} = \mathcal{D}_{\kappa}^{c} \rightarrow$ apply (post-)boostIV$(\mathcal{T}_{\kappa},\,\mathcal{M}[i])$ $\rightarrow g_{\mathcal{M}[i],\kappa}^{boost}(x)$
					\item[2.] validation set $\mathcal{V}_{\kappa} = \mathcal{D}_{\kappa} \rightarrow CV^{err}_{\kappa}(\mathcal{M}[i]) = \frac{1}{|\mathcal{V}_{\kappa}|}\sum_{i \in \mathcal{V}_{\kappa}}\left(y_{i} - g_{\mathcal{M}[i],\kappa}^{boost}(x_{i})\right)^{2}$ 
				\end{enumerate}
			} %
			calculate $CV^{err}(\mathcal{M}[i]) = \frac{1}{k}\sum_{\kappa=1}^{k}CV^{err}_{\kappa}(\mathcal{M}[i])$ \\
			\eIf(\tcp*[r]{Early stopping criterion}) {$CV^{err}(\mathcal{M}[i])$ > $CV^{err}(\mathcal{M}[i-1]) + \epsilon$} {
				$M^{*} = \mathcal{M}[i-1]$ \\
				\Break \tcp*[r]{Break while loop if the criterion is met}
			}{$i =+ 1$
			} %
		 } %
		 $M^{*} = \bar{M}$ \\
		 \KwOut{$g^{boost}_{M^{*}}(x) \leftarrow$ (post-)boostIV$(\mathcal{D},\,M^{*})$}
	\label{algo:CV_ES}
\end{algorithm}

\section{Theoretical properties} \label{sec:consistency}

In this section, we show that under mild conditions boostIV is consistent. Theoretical properties of post-boostIV are beyond the scope of the paper and are left for future research.

We borrow the main idea from \cite{zhang_yu2005boostingES} and modify it accordingly to apply it to the GMM criterion. Let $g(W_{i}, f) = (y_{i} - f(x_i))z_{i}$ denote a $k \times 1$ moment function, then $g_{0}(f) = \mathbb{E}[g(W_{i}, f)]$ is the population moment function and $\hat{g}(f) = n^{-1}\sum_{i=1}^{n}g(W_{i}, f)$ is its sample analog. Also let $\Omega$ denote a $k \times k$ positive semi-definite weight matrix and $\hat{\Omega}$ be its sample analog. Thus, the population GMM criterion and its sample analog are 
\begin{equation} \label{eq:gmm_crit}
	Q(f) = g_{0}(f)'\Omega g_{0}(f), \quad \hat{Q}(f) = \hat{g}(f)'\hat{\Omega}\hat{g}(f).
\end{equation}
The form of the GMM criterion in \eqref{eq:gmm_crit} corresponds to the form of the empirical objective function in \cite{zhang_yu2005boostingES} with the loss function replaced by the moment function.

We follow \cite{zhang_yu2005boostingES} and replace the functional gradient decent step \eqref{eq:boostIV_iteration} leading to the 2SLS fitting procedure on every iteration with an approximate minimization involving a GMM criterion. We can do that since the 2SLS solution is a special case of a GMM solution with an appropriate weighting matrix.

\begin{assumption}
    \textbf{Approximate Minimization.} On each iteration step $m$ we find $\bar{\alpha}_{m} \in \Lambda_{m}$ and $\bar{g}_{m} \in \mathcal{S}$ such that
    \begin{equation} \label{eq:approx_min }
        \hat{Q}(f_{m} + \bar{\alpha}_{m}\bar{g}_{m}) \le \inf_{\alpha_{m}\in \Lambda_m g_{m} \in \mathcal{S}}\hat{Q}(f_{m} + \alpha_{m}g_{m}) + \epsilon_m,
    \end{equation}
    where $\epsilon_m$ is a sequence of non-negative numbers that converge to 0.
    \label{ass_1}
\end{assumption}

As \cite{zhang_yu2005boostingES} show, the consistency of the boosting procedure consists of two parts: (i) numerical convergence of the procedure itself, i.e. the algorithm achieves the true minimum of the objective function, and (ii) statistical convergence that insures the uniform convergence of the sample criterion to its population analog. We will treat these two steps separately in the following subsections, and then combine them to demonstrate consistency of the boostIV.

\subsection{Numerical Convergence}

To demonstrate numerical convergence, we first have to verify that the sample GMM criterion in \eqref{eq:gmm_crit} satisfies Assumption 3.1 from \cite{zhang_yu2005boostingES}.

Following \cite{zhang_yu2005boostingES}, I introduce some additional notation. Let $\mathcal{S}$ be a set of real-valued functions and define
\begin{equation*}
	\text{span}(\mathcal{S}) = \left\{ \sum_{j=1}^{J}w_{j}f_{j}:f_{j} \in \mathcal{S},\,w_{i}\in\mathbb{R},\,J\in\mathbb{Z}^{+} \right\},
\end{equation*}
which forms a linear function space. Also, for all $f \in \text{span}(S)$ define the 1-norm with respect to the basis $S$ as
\begin{equation*}
	||f||_{1} = \left\{ ||w||_{1}: f = \sum_{j=1}^{J}w_{j}f_{j}:f_{j} \in \mathcal{S},\,J\in\mathbb{Z}^{+} \right\}.
\end{equation*}

\begin{assumption} \label{ass:ass_crit}
	A convex function $A(f)$ defined on $\text{span}(\mathcal{S})$ should satisfy the following conditions:
\begin{enumerate}
	\item The functional $A$ satisfies the following Frechet-like differentiability condition
	\begin{equation*}
		\lim_{h \rightarrow 0} \frac{1}{h}(A(f + h\varphi) - A(f)) = \nabla A'\varphi
	\end{equation*}
	\item For all $f\in \text{span}(\mathcal{S})$ and $\varphi \in \mathcal{S}$, the real-valued function $A_{f,\varphi}(h) = A(f + h\varphi)$ is second-order differentiable (as a function of $h$) and the second derivative satisfies
	\begin{equation*}
		A''_{f,\,\varphi}(0) \leq M(||f||_{1}),
	\end{equation*}
	where $M(\cdot)$ is a nondecreasing real-valued function.
\end{enumerate}

\end{assumption}

\begin{lemma} \label{lemma:gmm_crit}
	Let (i) the basis functions $\varphi$ be bounded as $\sup_{x}|\varphi(x)^{2}|=C<\infty$, (ii) the maximal eigenvalue $\lambda_{max}$ of the weighting matrix $\Omega$ be bounded from above, $\lambda_{max}(\Omega) < \infty$, and (iii) $\mathbb{E}[|z_{i}'z_{i}|] \leq B < \infty$. Then the population GMM criterion defined in \eqref{eq:gmm_crit} satisfies Assumption \ref{ass:ass_crit}.
\end{lemma} 

\begin{assumption} \label{ass:step_size} 
\textbf{Step size.}
\begin{itemize}
	\item[(a)] Let $\Lambda_{m} \subset \mathbb{R}$ such that $0 \in \Lambda_{m}$ and $\Lambda_{m} = -\Lambda_{m}$.
	\item[(b)] Let $h_{m} = \sup \Lambda_{m}$ satisfy the conditions
	\begin{equation}
		\sum_{j=0}^{\infty}h_{j} = \infty, \quad \sum_{j=0}^{\infty}h_{j}	^{2}<\infty.
	\end{equation}
\end{itemize}

Then we can bound the step size $|\bar{\alpha}_{m}| \leq h_{m}$.
\end{assumption}

Note that Assumption \ref{ass:step_size}(a) restricts the step size $\alpha_{m}$. \cite{friedman2001greedy} argues that restricting the step size is always preferable in practice, thus, we will restrict our attention to this case\footnote{\cite{zhang_yu2005boostingES} provide a short discussion on how to deal with the unrestricted step size, however, the argument relies on the exact minimization which greatly complicates the analysis.}. Moreover, $\Lambda_m$ is allowed to depend on the previous steps of the algorithm. Assumption \ref{ass:step_size}(b) requires the step size $h_j$ to be small ($\sum_{j=0}^{\infty}h_{j}^{2}<\infty$) preventing large oscillation, but not too small ($\sum_{j=0}^{\infty}h_{j} = \infty$) ensuring that $f_{m}$ can cover the whole $\text{span}(\mathcal{S})$. The following theorem establishes the main numerical convergence result.

\begin{theorem} \label{thm:num_convergence}
	Assume that we choose quantities $f_{0}$, $\epsilon_{m}$ and $\Lambda_{m}$ independent of the sample $W$. Given the results of Lemma \ref{lemma:gmm_crit}, as long as there exists $h_j$ satisfying Assumption \ref{ass:step_size} and $\epsilon_{j}$ such that $\sum_{j=0}^{\infty}\epsilon_{j} < \infty$, we have the following convergence result:
	\begin{equation*}
		\lim_{m\rightarrow\infty}\hat{Q}(f_{m}) = \inf_{f\in\text{span}(\mathcal{S})}\hat{Q}(f).
	\end{equation*}
\end{theorem}

\subsection{Statistical convergence}

We need to show that the sample GMM criterion uniformly converges to its population analog, then under proper regularity conditions we will be able to ensure consistency of boostIV.

To show that the sample GMM criterion converges uniformly to its population analog, we will first bound the moment function and then we will show that it is sufficient to put a bound on the criterion function. 

\begin{assumption} \label{ass:uc}
	Assume the following conditions:
	\begin{enumerate}
		\item The class of weak learners $\mathcal{S}$ is closed under negation, i.e. $f \in \mathcal{S} \rightarrow -f \in \mathcal{S}$.
		\item The moment function is Lipschitz with each component $j = 1,\,\dots,\,k$ satisfying
		\begin{equation*}
			\exists \gamma_{j}(\beta) \text{ such that } \forall |f_{1}|,\,|f_{2}| \leq \beta \quad |g_{j}(f_{1}) - g_{j}(f_{2})| \leq \gamma_{j}(\beta)|f_{1} - f_{2}|,
		\end{equation*}
		implying that
		\begin{equation*}
			||g(f_{1}) - g(f_{2})|| \leq \gamma(\beta)|f_{1} - f_{2}|, \quad \gamma(\beta) = \sqrt{\sum_{j=1}^{k}\gamma_{j}^{2}(\beta)}.
		\end{equation*}
	\end{enumerate}
\end{assumption}

To bound the rate of uniform convergence of the moment function, we  appeal to the concept of Rademacher complexity. Let $\mathcal{H} = {h(w)}$ be a set of real-valued functions. Let $\{\zeta_{i}\}_{i=1}^{n}$ be a sequence of binary random variables such that $\zeta_{i}$ takes values in $\{-1,1\}$ with equal probabilities. Then the sample or empirical Rademacher complexity of class $\mathcal{H}$ is given by
\begin{equation} \label{eq:rad_def}
	\hat{R}(\mathcal{H}) = \mathbb{E}_{\zeta}\left[\sup_{h \in \mathcal{H}} n^{-1}\sum_{i=1}^{n}\zeta_{i}h(W_{i})\right].
\end{equation}
We also denote $R(\mathcal{H}) = \mathbb{E}_{W}\hat{R}(\mathcal{H})$ to be the expected Rademacher complexity, where $\mathbb{E}_{W}$ is the expectation with respect to the sample $W = (W_{1},\,\dots,\,W_{n})$. Note that the definition in \eqref{eq:rad_def} differs from the standard definition of Rademacher complexity where there is an absolute value under the supremum sign \parencite{vaart_wellner1996weak}. The current version of Rademacher complexity has the merit that it vanishes for function classes consisting of single constant function, and is always dominated by the standard Rademacher complexity. Both definitions agree for function classes which are closed under negation \parencite{meir_zhang2003generalization}.

\begin{lemma} \label{lemma:moment_uc}
	Under Assumption \ref{ass:uc}, for all $j = 1,\,\dots,\,k$,
	\begin{equation*}
		\mathbb{E}_{W} \sup_{||f||_{1} \leq \beta} |g_{0,j}(f) - \hat{g}_{j}(f)| \leq 2\gamma_{j}(\beta)\beta R(\mathcal{S}).
	\end{equation*}
\end{lemma}

For many classes the Rademacher complexity can be calculated directly, however, to obtain a more general result we need to bound $R(\mathcal{S})$. Using the results from Section 4.3 in \cite{zhang_yu2005boostingES} we can bound the expected Rademacher complexity of the weak learner class by
\begin{equation} \label{eq:rad_bound}
	R(\mathcal{S}) \leq \frac{C_{\mathcal{S}}}{\sqrt{n}},
\end{equation}
where $C_{\mathcal{S}}$ is a constant that solely depends on $\mathcal{S}$. \cite{zhang_yu2005boostingES} also show that popular weak learners such as two-level neural networks and trees basis functions satisfy the requirements. However, \cite{zhang_yu2005boostingES} point out that in general the bound may be slower than root-n. In Appendix \ref{app:rad_bound} we derive an alternative bound on $R(\mathcal{S})$ that works for any class with finite VC dimension. The derived VC bound is slower by the factor of $\log(n)$ that appears in a lot of ML algorithms. 

Condition \eqref{eq:rad_bound} allows us to bound the moment function which leads to a bound on the rate of uniform convergence of the GMM criterion. The formal statements of the results are presented below.

\begin{lemma} \label{lemma:moment_consistency}
	Suppose that condition \eqref{eq:rad_bound} holds, then under Assumption \ref{ass:uc}, 
	\begin{equation*}
		\sup_{||f||_{1} \leq \beta} ||g_{0}(f) - \hat{g}(f)|| \overset{p}{\to} 0.
	\end{equation*}
\end{lemma}

\begin{theorem} \label{thm:criterion_uc}
	Suppose that (i) the data $W = (W_{1},\,\dots,\,W_{n})$ are i.i.d., (ii) $\hat{\Omega} \overset{p}{\to} \Omega$, (iii) Assumption \ref{ass:uc} is satisfied, and (iv) $\mathbb{E}_{W}\left[\sup_{||f||_{1} \leq \beta}||g(W_{i}, f)||\right] < \infty$. Then
	\begin{equation*}
		\sup_{||f||_{1} \leq \beta}|\hat{Q}(f) - Q(f)| \overset{p}{\to} 0.
	\end{equation*} 
\end{theorem}

\subsection{Consistency}

In this section we put together the arguments for numerical and statistical convergence presented in the previous subsections to prove consistency of the boostIV algorithm. We start with a general decomposition illustrating the proof strategy and highlighting where exactly numerical and statistical convergence step in.

Suppose that we run the boostIV algorithm and stop at an early stopping point $\hat{m}$ that satisfies $\mathbb{P}(||\hat{f}_{\hat{m}}||_{1} \leq \beta_{n}) = 1$ for some sample-independent $\beta_{n} \geq 0$. Let $f^{*}$ be a unique minimizer of the population criterion, i.e. $Q(f^{*}) = \inf_{f \in \text{span}(\mathcal{S})}Q(f)$. By the triangle inequality, we get the following decomposition
\begin{align*}
	\left| Q(\hat{f}_{\hat{m}}) - Q(f^{*}) \right| & \leq \left| Q(\hat{f}_{\hat{m}}) - \hat{Q}(\hat{f}_{\hat{m}}) \right| + \left| \hat{Q}(\hat{f}_{\hat{m}}) - \hat{Q}(f^{*}) \right| + \left| \hat{Q}(f^{*}) - Q(f^{*}) \right| \\
	& \leq 2\sup_{||f||_{1}\leq\beta}\left|\hat{Q}(f) - Q(f)\right| + \left| \hat{Q}(\hat{f}_{\hat{m}}) - \hat{Q}(f^{*}) \right| 
\end{align*}
We can bound the first term using the uniform bound on the sample GMM criterion in Theorem \ref{thm:criterion_uc}, this is the statistical convergence argument. In order to bound the second term, we have to appeal to the numerical convergence argument in Theorem \ref{thm:num_convergence}. As a result, since $Q(\hat{f}_{\hat{m}}) \rightarrow Q(f^{*})$ as $n\rightarrow\infty$, it follows that $\hat{f}_{\hat{m}} \overset{p}{\to} f^{*}$. The following theorem formalizes the result.

\begin{theorem} \label{thm:consistency}
	Suppose that the assumptions of Theorems \ref{thm:num_convergence} and \ref{thm:criterion_uc} hold. Consider two sequences $k_{n}$ and $\beta_{n}$ such that $\lim_{n\rightarrow\infty}m_{n}=\infty$ and $\lim_{n\rightarrow\infty}\gamma(\beta_{n})\beta_{n}R(\mathcal{S}) = 0$. Then as long as we stop the algorithm at step $\hat{m}$ based on $W$ such that $\hat{m} \geq m_{n}$ and $||\hat{f}_{\hat{m}}||_{1} \leq \beta_{n}$, we have the consistency result $\hat{f}_{\hat{m}} \overset{p}{\to} f^{*}$.
\end{theorem}

\section{Monte Carlo experiments} \label{sec:mc}

\subsection{Univariate design}

To begin with, we consider a simple low-dimensional scenario with one endogenous variable and two instruments. 
\begin{equation*}
	y = g(x) + \rho e + \delta, \quad x = z_{1} + z_{2} + e + \gamma,
\end{equation*}
where instruments $z_{j} \sim U[-3,\,3]$ for $j = 1,\,2$, $e \sim \mathcal{N}(0,\,1)$ is the confounder, $\delta,\,\gamma \sim \mathcal{N}(0,\,0.1)$ are additional noise components, and $\rho$ is the parameter measuring the degree of endogeneity, which we set to 0.5 in the simulations. We focus on four specifications of the structural function:
\begin{itemize}
	\item \textbf{abs:} $g(x) = |x|$
	\item \textbf{log:} $g(x) = \log(|16x - 8| + 1)\text{sign}(x - 0.5)$
	\item \textbf{sin:} $g(x) = \sin(x)$
	\item \textbf{step:} $g(x) = \mathds{1}\{x < 0\} + 2.5\times\mathds{1}\{x \geq 0\}$
\end{itemize}

We compare the performance of boostIV and post-boostIV with the standard NPIV estimator using the cubic polynomial basis, Kernel IV (KIV) regression of \cite{singh2019kiv}\footnote{Code: \url{https://github.com/r4hu1-5in9h/KIV}}, DeepIV estimator of \cite{hartford2017}\footnote{We use the latest implementation of the econML package: \url{https://github.com/microsoft/EconML}} and DeepGMM estimator of \cite{bennett2019deepGMM}\footnote{Code: \url{https://github.com/CausalML/DeepGMM}}. We use $1,000$ observations for both train and test sets and $500$ observations for the validation set. Our results are based on $200$ simulations for each scenario.

\begin{sidewaysfigure}
\hspace*{-1cm}
	\centering
	\includegraphics[scale=0.8]{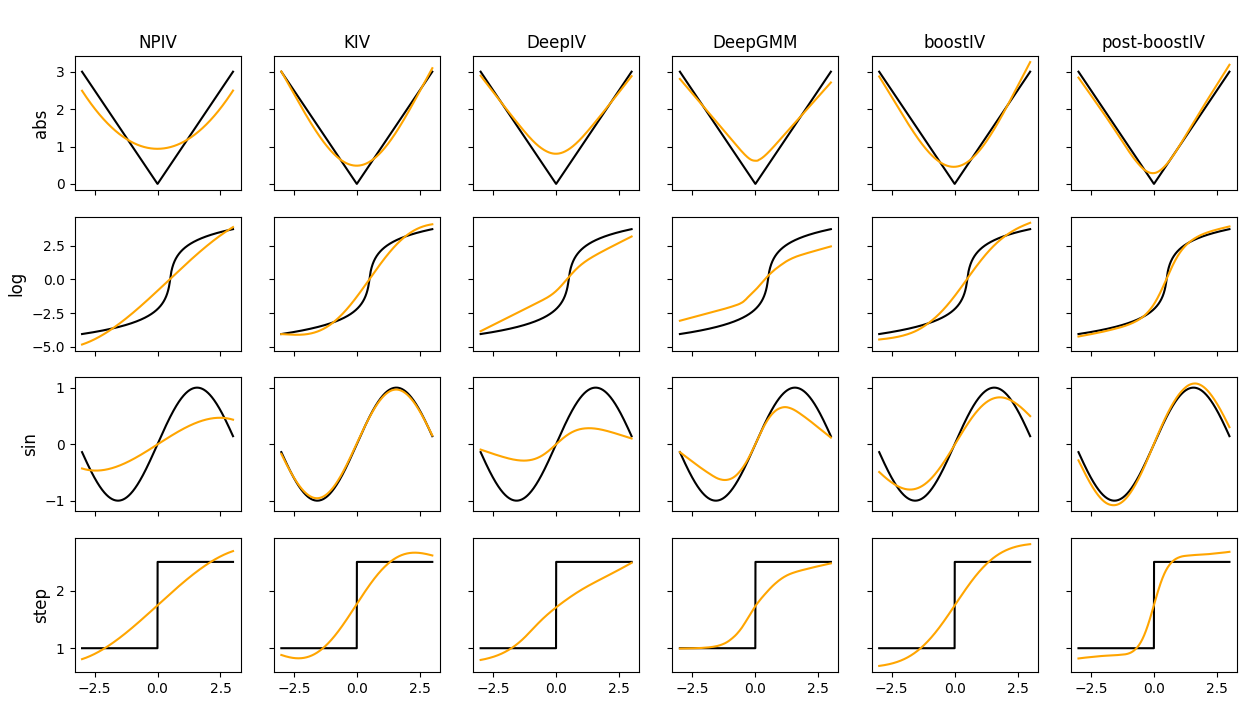}
	\caption{\textnormal{Out-of-sample average fit across simulations. The black line is the true function, the orange line is the fit.}}
	\label{fig:univar_fit}
\end{sidewaysfigure}

\begin{table}[!h] 
\centering
\caption{Univariate design: Out-of-sample MSE.}
\begin{tabular}{l|cccccc}
\toprule
 & NPIV & KIV & DeepIV & DeepGMM & boostIV & post-boostIV \\
\midrule
    abs & 0.1916 &  0.0564 &     0.1347 &      1.2717 &      0.0348 &           0.0217 \\
     log & 0.6936 &  0.3367 &     1.2708 &     14.4615 &      0.3173 &            0.0930 \\
     sin & 0.1837 &  0.0217 &     0.2798 &      0.8595 &      0.0292 &           0.0124 \\
     step & 0.1267 &  0.0972 &     0.1756 &      0.9796 &      0.1027 &           0.0546 \\
\bottomrule
\end{tabular}
\label{tab:univar}
\end{table}

We plot our results in Figures \ref{fig:univar_fit} which shows the average out of sample fit across simulations (orange line) compared to the true target function (black line). Table \ref{tab:univar} presents the out-of-sample MSE across simulations. First thing to notice is that NPIV fails to capture different functional form subtleties. Second, DeepIV's performance does not improve upon the one of NPIV. Moreover, even though DeepGMM estimates have lower bias than the ones of NPIV and DeepIV (except for the log function), they are quite volatile across simulations leading to higher MSE. BoostIV performs on par with KIV both in terms of the bias term as they are able to recover the underlying structural relation, and in terms of the variance leading to low MSE. Finally, the post-processing step helps to further improve upon boostIV's performance by reducing bias. On top of that, post-boostIV requires less iterations to converge. We use $5,000$ iterations for boostIV, while post-boostIV uses on average $50$ iterations.\footnote{In this experiment we do not tune boostIV, we just pick a large enough number of iterations for it to converge. However, we do tune post-boostIV.}

\subsection{Multivariate Design}

Consider the following data generating process:
\begin{align*}
	y_{i} & = h(x_{i}) + \varepsilon_{i} \\
	x_{i,k} & = g_{k}(z_{i}) + v_{i,k}, \quad k = 1,\,\dots,\,d_{x},
\end{align*}
where $y_{i} \in \mathbb{R}$ is the response variable, $x_{i} \in \mathbb{R}^{d_{x}}$ is the vector of potentially endogenous variables, $z_{i} \in \mathbb{R}^{d_{z}}$ is the vector of instruments, $\varepsilon_{i} \in \mathbb{R}$ is the structural error term, and $v_{i} \in \mathbb{R}^{d_{x}}$ is the vector of the reduced form errors. Function $h(\cdot)$ is the structural function of interest, and function $g(\cdot)$ governs the reduced form relationship between the endogenous regressors and instrumental variables.

Instruments are drawn from a multivariate normal distribution, $z_{i} \sim \mathcal{N}(0, \Sigma_{z})$, where $\Sigma_{z}$ is just an identity matrix. The error terms are described by the following relationship:
\begin{equation*}
	\varepsilon \sim \mathcal{N}(0,\,1), \quad v \sim \mathcal{N}(\rho \varepsilon,\,\mathcal{I} - \rho^{2}),
\end{equation*}
where $\rho$ is the correlation between $\varepsilon$ and the elements of $v$, which controls the degree of endogeneity.

We consider two structural function specifications: 
\begin{enumerate}
	\item a simpler design where the structural function is proportional to a multivariate normal density, i.e. $h(x) = \exp\{-0.5 x'x\}$. We will further refer to this specification as Design 1;
	\item a more challenging design where the structural function is $h(x) = \sum_{k=1}^{d_x}\sin(10x_{k})$. We will further refer to this specification as Design 2.
\end{enumerate} 
We also consider two different choices of the reduced form function $g(\cdot)$:
\begin{itemize}
	\item [(a)] linear: $g(Z_{i}) = Z_{i}'\Pi$, where $\Pi \in \mathbb{R}^{d_x \times d_z}$ is a matrix of reduced form parameters;
	\item [(b)] non-linear: $g_{k}(Z_{i}) = G(Z_{i}; \theta_{k})$ for $k = 1,\dots,d_{x}$, where $G(Z_{i}; \theta_{k})$ is a multivariate normal density parameterized by the mean vector $\theta_{k}$ (for simplicity, we use the identity covariance matrix).
\end{itemize}

We use $1,000$ observations for the train set and $500$ observations for both the validation and test sets. We run $200$ simulations for each scenario. The results are summarized in Tables \ref{tab:m_design_1} and \ref{tab:m_design_2}.

\begin{table}[!h] 
\centering
\caption{Design 1: Out-of-sample MSE.}
\begin{tabular}{llll|cccccc}
\toprule
dx &  dz &  IV type & $\rho$ & NPIV & KIV & DeepIV & DeepGMM & boostIV & post-boostIV \\
\midrule
    5 &   7 &     lin &  0.25 &   4.9535 &  0.0147 &     0.0497 &      0.2234 &      0.0213 &           1.3306 \\
     &    &      &  0.75 &   6.8889 &  0.0286 &      0.054 &      0.1655 &      0.0603 &           0.7249 \\
     &    &  nonlin &  0.25 &   4.0548 &   0.017 &     0.0757 &      0.3262 &      0.0875 &           0.3457 \\
     &    &   &  0.75 &   1.9932 &  0.0516 &     0.1188 &      0.7265 &      0.4287 &           0.9128 \\
\midrule
    10 &  12 &     lin &  0.25 &  23.1025 &  0.0024 &     0.0867 &      0.3089 &      0.0084 &           1.0953 \\
     &   &      &  0.75 &  39.6902 &  0.0108 &     0.0884 &       0.251 &      0.0427 &           0.6347 \\
     &   &  nonlin &  0.25 &   6.4842 &  0.0038 &       0.05 &      0.4908 &      0.0525 &           0.8137 \\
     &   &   &  0.75 &     2.53 &  0.0147 &     0.0691 &       0.822 &      0.3332 &           0.8937 \\
\bottomrule
\end{tabular}
\label{tab:m_design_1}
\end{table}

\begin{table}[!h] 
\centering
\caption{Design 2: Out-of-sample MSE.}
\begin{tabular}{llll|cccccc}
\toprule
dx &  dz &  IV type & $\rho$ & NPIV & KIV & DeepIV & DeepGMM & boostIV & post-boostIV \\
\midrule
	5 &   7 &     lin &  0.25 &  21.5854 &  2.4983 &     2.5484 &      2.9105 &      2.5105 &           3.5498 \\
    &   &      &  0.75 &  23.2413 &  2.5043 &     2.5358 &       2.792 &      2.5351 &           3.5081 \\
    &   &  nonlin &  0.25 &  19.0871 &  2.5118 &     2.5415 &      2.9867 &      2.5707 &           3.0303 \\
    &   &   &  0.75 &  22.1192 &  2.5367 &     2.5523 &      3.4188 &       2.891 &           3.7043 \\
\midrule
    10 &  12 &     lin &  0.25 &  147.984 &  5.0047 &     5.1383 &      5.9647 &      5.0209 &           5.9435 \\
    &   &      &  0.75 &   241.56 &  5.0326 &     5.1698 &      5.7147 &      5.0781 &           6.3259 \\
    &   &  nonlin &  0.25 &  61.0328 &  5.0103 &     5.0636 &      6.2145 &      5.0713 &           5.9001 \\
    &   &   &  0.75 &  112.785 &  4.9799 &     5.0631 &       6.193 &      5.3172 &           6.4674 \\
\bottomrule
\end{tabular}
\label{tab:m_design_2}
\end{table}

\subsection{Application to nonparametric demand estimation}

In this section we apply our algorithms to a more economically driven example of demand estimation. Demand estimation is a cornerstone of modern industrial organization and marketing research. Besides its practical importance, it poses a challenging estimation problem  which modern econometric and statistical tools can be applied to.

We consider a nonparametric demand estimation framework of \cite{gandhi2020flexible} (hereafter, GNT). GNT is a flexible framework that combines the nonparametric identification arguments of \cite{berry2014identification} with the dimensionality reduction techniques of \cite{gandhi_houde2019}.

In market $t$, $t=1,\dots,T$, there is a continuum of consumers choosing from a set of products $\mathcal{J} = \{0,1,\dots,J\}$ which includes the outside option, e.g. not buying any good. The choice set in market $t$ is characterized by a set of product characteristics $\chi_{t}$ partitioned as follows:
\begin{equation*}
	\chi_{t} \equiv (x_{t}, p_{t}, \xi_{t}),
\end{equation*}
where $x_{t} \equiv (x_{1t}, \dots, x_{Jt})$ is a vector of exogenous observable characteristics (e.g. exogenous product characteristics or market-level income), $p_{t} \equiv (p_{1t}, \dots, p_{Jt})$ are observable endogenous characteristics (typically, market prices) and $\xi_{t} \equiv (\xi_{1t}, \dots, \xi_{Jt})$ represent unobservables potentially correlated with $p_{t}$ (e.g. unobserved product quality). Let $\mathcal{X}$ denote the support of $\chi_{t}$. Then the structural demand system is given by
\begin{equation*}
	\sigma: \mathcal{X} \mapsto \Delta^{J},
\end{equation*}
where $\Delta^{J}$ is a unit $J$-simplex. The function $\sigma$ gives, for every market $t$, the vector $s_{t}$ of shares for the $J$ goods.

Following \cite{berry2014identification}, we partition the exogenous characteristics as $x_{t} = \left(x_{t}^{(1)}, x_{t}^{(2)} \right)$, where $x_{t}^{(1)} \equiv \left(x_{1t}^{(1)}, \dots, x_{Jt}^{(1)}\right)$, $x_{jt} \in \mathbb{R}$ for $j \in \mathcal{J} \backslash \{0\}$, and define the linear indices
\begin{equation*}
	\delta_{jt} = x_{jt}^{(1)}\beta_{j} + \xi_{jt}, \quad j \in \mathcal{J} \backslash \{0\},
\end{equation*}
and let $\delta_{t} \equiv (\delta_{1t}, \dots, \delta_{Jt})$. Without loss of generality, we can normalize $\beta_{j} = 1$ for all $j$ (see \cite{berry2014identification} for more details). Given the definition of the demand system, for every market $t$
\begin{equation*}
	\sigma(\chi_{t}) = \sigma\left(\delta_{t}, p_{t}, x_{t}^{(2)}\right).
\end{equation*}

Following \cite{berry2013connected} and \cite{berry2014identification}, we can show that there exists at most one vector $\delta_{t}$ such that $s_{t} = \sigma\left(\delta_{t}, p_{t}, x_{t}^{(2)}\right)$, meaning that we can write
\begin{equation} \label{eq:inv_demand}
	\delta_{jt} = \sigma_{j}^{-1}\left(s_{t}, p_{t}, x_{t}^{(2)} \right), \quad j \in \mathcal{J} \backslash \{0\}.
\end{equation}
We can rewrite \eqref{eq:inv_demand} in a more convenient form to get the following estimation equation 
\begin{equation} \label{eq:est_eq1}
	x_{jt}^{(1)} = \sigma_{j}^{-1}\left(s_{t}, p_{t}, x_{t}^{(2)} \right) - \xi_{jt}.
\end{equation}

Note that in \eqref{eq:est_eq1} the inverse demand is indexed by $j$, meaning that we have to estimate $J$ inverse demand functions. To circumvent this problem, \cite{gandhi_houde2019} suggest transforming the input vector space under the linear utility specification to get rid of the $j$ subscript. GNT follow this idea and show that Equation \eqref{eq:est_eq1} can be rewritten as
\begin{equation} \label{eq:est_eq2}
	\log\left(\frac{s_{jt}}{s_{0t}}\right) = x_{jt}^{(1)} + g(\omega_{jt}) + \xi_{jt},
\end{equation}
where $g$ is such that 
\begin{equation*}
	\sigma_{j}^{-1}\left(s_t, p_t, x_{t}^{(2)}\right) = \left(\frac{s_{jt}}{s_{0t}}\right) - g(\omega_{jt}),
\end{equation*}
and $\omega_{jt} \equiv \left(s_{jt}, \{s_{kt}, d_{jkt}\}_{j \neq k}\right)$, where $d_{jkt} = \tilde{x}_{jt} - \tilde{x}_{kt}$ and $\tilde{x}_{t} \equiv \left(p_t, x_{t}^{(2)}\right)$. 

Let $y_{jt} \equiv \log(s_{jt} / s_{0t}) - x_{jt}^{(1)}$, then we can rewrite equation \eqref{eq:est_eq2} in a more convenient form
\begin{equation} \label{eq:demand_main_eq}
	y_{jt} = g(\omega_{jt}) + \xi_{jt}.
\end{equation}
Thus, \eqref{eq:demand_main_eq} is our structural equation where $\omega_{jt}$ contains endogenous variables. Assume we have a cost-shifter $w_{jt}$ that is exogenous, then given $\mathbb{E}[\xi_{jt}|x_{t}, w_{t}] = 0$ for $j \in \mathcal{J}\backslash\{0\}$\footnote{We also need the completeness condition to be satisfied, see \cite{berry2014identification} for more details.} we can estimate the inverse demand function $g$. To constructs instruments, GNT transform the input space $(x_{t},\,w_{t})$ similarly to $\omega_{jt}$. Let $\zeta_{jt} \equiv  \{\Delta_{jkt}\}_{j\not=k}$, where $\Delta_{jkt} = z_{jt} - z_{kt}$ and $z_{t} = (x_{t},\,w_{t})$. Thus, we can perform estimation based on $\mathbb{E}[\xi_{jt}|\zeta_{jt}] = 0$.

\begin{table}[!h]
\centering
\caption{Inverse demand fit.}
\begin{tabular}{lll|cccc}
\toprule
 &  &  & KIV & DeepGMM & boostIV & post-boostIV \\
\hline
T & J & K & \multicolumn{4}{c}{Bias} \\
\midrule
100 & 10 & 10 &    5.3307 &        0.9852 &        0.1917 &             0.9847 \\	
	& 	 & 20 &    6.5053 &        1.6591 &       -0.3053 &             0.9100 \\
	& 20 & 10 &    5.4312 &        1.6979 &        0.1016 &             0.8351 \\
	&	 & 20 &  6.6666 &        3.2286 &       -0.4047 &             0.9389 \\
\midrule
T & J & K & \multicolumn{4}{c}{MSE} \\
\hline
100 & 10 & 10 &  45.5991 &      36.4565 &      15.1358 &            7.5426 \\
	& 	 & 20 &  73.2391 &      62.9117 &      26.1670 &           15.5126 \\
	& 20 & 10 &  47.4954 &      40.8320 &      15.8077 &            6.5199 \\
	&	 & 20 &  76.8815 &      70.0830 &      27.4539 &           14.2017 \\
\bottomrule
\end{tabular}
\label{tab:nonparam_demand}
\end{table}

In our model design there are $T=100$ markets with $J \in \{10,\,20\}$ products and $K \in \{10,\, 20\}$ nonlinear characteristics besides the price. We compare the performance of boostIV and post-boostIV with KIV and DeepGMM. We drop NPIV since in our design it fails due to the curse of dimensionality. We also drop DeepIV as it suffers from the exploding gradient problem. 

Table \ref{tab:nonparam_demand} summarizes the results. The first thing to notice is that KIV performs the worst, while in the previous experiments it was one of the best performing estimators. It has both high bias and variance. DeepGMM has smaller bias, but the variance is still big. Our algorithms clearly dominate KIV and GMM, with post-boostIV delivering the best MSE results, while having slightly higher bias compared to boostIV.

\section{Conclusion} \label{sec:conclusion}

In this paper we have introduced a new boosting algorithm called boostIV that allows to learn the target function in the presence of endogenous regressors. The algorithm is very intuitive as it resembles an iterative version of the standard 2SLS regression. We also study several extensions including the use of optimal instruments and the post-processing step.

We show that boostIV is consistent and demonstrates an outstanding finite sample performance in the series of Monte Carlo experiments. It performs especially well in the nonparameteric demand estimation example which is characterized by a complex nonlinear relationship between the target function and explanatory features.

Despite all the advantages of boostIV, the algorithm does not allow for high-dimensional settings where the number of regressors and/or instruments exceeds the number of iterations. We also believe it is possible to extend our algorithm in the spirit similar to XGBoost \citep{chen2016xgboost} that could decrease the computation time taken by the algorithm. These would be interesting directions for future research.

\newpage

\begin{appendices}

\section{Auxiliary Lemmas}

\setcounter{lemma}{0}
\renewcommand{\thelemma}{\Alph{section}\arabic{lemma}}
	
\begin{lemma} \label{lemma:delta_bound}
    Assume the assumptions of Lemma \ref{lemma:gmm_crit} are satisfied. Consider $h_m$ that satisfies Assumption \ref{ass:step_size}. Let $\bar{f}$ be an arbitrary reference function in $\mathcal{S}$. Also, define $s_m = ||f_0||_1 + \sum_{i=0}^{m-1}h_i$, and
    \begin{equation}
        \Delta \hat{Q}(f) = \max\left(0, \hat{Q}(f) - \hat{Q}(\bar{f})\right),
    \end{equation}
    \begin{equation}
        \bar{\epsilon}_{m} = \frac{h^{2}_{m}}{2}M + \epsilon_{m}.
    \end{equation}
    Then after $m$ steps the following bound holds for $f_{m+1}$:
    \begin{equation}
        \Delta \hat{Q}(f_{k+1}) \le \left(1 - \frac{h_m}{s_m + ||\bar{f}||_1}\right)\Delta  \hat{Q}(f_{m}) + \bar{\epsilon}_{m}
    \end{equation}
\end{lemma}
\begin{proof}
    The result follows directly from Lemma \ref{lemma:gmm_crit} and  Lemma 4.1 in \cite{zhang_yu2005boostingES}. 
\end{proof}

\begin{lemma} \label{lemma:delta_bound_2}
    Under the assumptions of Lemma \ref{lemma:delta_bound}, we have
    \begin{equation}
    \Delta \hat{Q}(f_{m}) \leq \frac{||f_0||_1 + ||\bar{f}||_1}{s_m + ||\bar{f}||_1}\Delta \hat{Q}(f_{0}) + \sum_{j=1}^{m}\frac{s_j + ||\bar{f}||_1}{s_m + ||\bar{f}||_1}\bar{\epsilon}_{j-1}    
    \end{equation}
    \end{lemma}
\begin{proof}
The above lemma directly follows from the repetitive application of Lemma \ref{lemma:delta_bound}. For detailed proof see \cite{zhang_yu2005boostingES}.         
\end{proof}

Lemmas \ref{lemma:delta_bound} and \ref{lemma:delta_bound_2} are direct counterparts of Lemmas 4.1 and 4.2 in \cite{zhang_yu2005boostingES} with $M(s_{m+1})$ replaced by $M$. Therefore, the main numerical convergence result below follows as well (see Corollary 4.1).

\section{Proofs}

\subsection{Proof of Lemma \ref{lemma:gmm_crit}}

First, $Q(\cdot)$ is convex in $f$, hence, it is convex differentiable. Now we have to bound the second derivative with respect to $h$. Note that the second derivative of $Q_{f,\,\varphi}(h)$ does not even depend on $h$,
\begin{align*}
	Q''_{f,\,\varphi}(h) & = \mathbb{E}[\varphi(x_{i})z_{i}]'\Omega\mathbb{E}[\varphi(x_{i}) z_{i}] \\ 
	& \leq \lambda_{max}(\Omega)||\mathbb{E}[\varphi(x_{i}) z_{i}]||^{2} \\
	& \leq \lambda_{max}(\Omega) \mathbb{E}[|\varphi(x_{i})|^{2}]\mathbb{E}[|z_{i}'z_{i}|] \\ 
	& \leq \lambda_{max}(\Omega)CB \equiv M < \infty, 
\end{align*}
where the second inequality is a by the Cauchy-Schwarz inequality, and the last inequality comes from the assumptions of the lemma. Thus, the second derivative has a fixed bound $M < \infty$. \qed

\subsection{Proof of Theorem \ref{thm:num_convergence}}

The result follows directly from Lemmas \ref{lemma:delta_bound} and \ref{lemma:delta_bound_2}. For detailed proof see \cite{zhang_yu2005boostingES}. \qed

\subsection{Proof of Lemma \ref{lemma:moment_uc}}

Follows directly from Lemma 4.3 in \cite{zhang_yu2005boostingES}. \qed

\subsection{Proof of Lemma \ref{lemma:moment_consistency}}

It follows from Lemma \ref{lemma:moment_uc} and condition \eqref{eq:rad_bound} that for all $j = 1,\,\dots,\,k$,
\begin{equation*}
	\mathbb{E}_{W} \sup_{||f||_{1} \leq \beta} |g_{0,j}(f) - \hat{g}_{j}(f)| \leq 2\gamma_{j}(\beta)\beta R(\mathcal{S}) \leq 2\gamma_{j}(\beta)\beta\frac{C_{\mathcal{S}}}{\sqrt{n}} = O(n^{-1/2}).
\end{equation*}
Thus, by Markov inequality,
\begin{equation*}
	\sup_{||f||_{1} \leq \beta} |g_{0,j}(f) - \hat{g}_{j}(f)| \overset{p}{\to} 0, \quad j = 1,\,\dots,\,k.
\end{equation*}
Since every coordinate of the sample moment function converges uniformly to its population analog, we can bound the norm as well
\begin{equation*}
	||g_{0}(f) - \hat{g}(f)|| = \left(\sum_{j=1}^{k}|g_{0,j}(f) - \hat{g}_{j}(f)|^{2}\right)^{1/2} \leq \sqrt{k}O_{p}(n^{-1/2}),
\end{equation*}
which combined with Markov inequality completes the proof. \qed

\subsection{Proof of Theorem \ref{thm:criterion_uc}}

By the triangle and Cauchy-Schwarz inequalities,
\begin{align*}
		\left|\hat{Q}(f) - Q(f)\right| & \leq \left|[\hat{g}(f) - g_{0}(f)]'\hat{\Omega}[\hat{g}(f) - g_{0}(f)]\right| + \left|g_{0}(f)'(\hat{\Omega} + \hat{\Omega}')[\hat{g}(f) - g_{0}(f)]\right| \\
	& + \left|g_{0}(f)'(\hat{\Omega} - \Omega)g_{0}(f)\right| \\
	& \leq ||\hat{g}(f) - g_{0}(f)||^{2}||\hat{\Omega}|| + 2||g_{0}(f)||\,||\hat{g}(f) - g_{0}(f)||\,||\hat{\Omega}|| + ||g_{0}(f)||^{2}||\hat{\Omega} - \Omega||.
\end{align*}	
Using Lemma \ref{lemma:moment_consistency}, (ii), and (iv) and taking the supremum of both sides of the inequality completes the proof. \qed
	
\subsection{Proof of Theorem \ref{thm:consistency}}

By Theorem \ref{thm:criterion_uc}, the first term converges in probability to zero, and the second term converges to zero according to the arguments from the proof of Theorem 3.1 in \cite{zhang_yu2005boostingES}, which completes the proof. \qed

\section{Alternative bound on the Rademacher complexity} \label{app:rad_bound}

\setcounter{lemma}{0}
\renewcommand{\thelemma}{\Alph{section}\arabic{lemma}}

To derive an alternative bound on the Rademacher complexity, we introduce the following lemma (Massart's lemma).

\begin{lemma} \label{lemma:massart}
	For any $A \subseteq \mathbb{R}^{n}$, let $M = \sup_{a \in A}||a||$. Then
	\begin{equation*}
		\hat{R}(A) = \mathbb{E}_{\sigma}\left[\sup_{a \in A} \frac{1}{n}\sum_{i=1}^{n}\sigma_{i}a_{i} \right] \leq \frac{M\sqrt{2\log|A|}}{n}.
	\end{equation*}
\end{lemma}
This lemma can be applied to any finite class of functions.

\begin{exmp}
	Consider a set of binary classifiers $\mathcal{H} \subseteq \{h: W \mapsto \{-1,\,1\} \}$. Given a sample $W = (W_{1},\,\dots,\,W_{n})$, we can take $A = \{h(W_{1}),\,\dots,\,h(W_{n})\,|\,h\in\mathcal{H}\}$. Then $|A| = |\mathcal{H}|$ and $M = \sqrt{n}$. Massart's lemma gives
	\begin{equation*}
		\hat{R}(\mathcal{H}) \leq \sqrt{\frac{2\log|\mathcal{H}|}{n}}.
	\end{equation*} 
\end{exmp}

In general, Massart's lemma can also be applied to infinite function classes with a finite shattering coefficient. Notice that Massart's finite lemma places a bound on the empirical Rademacher complexity that depends only on $n$ data points. Therefore, all that matters as far as empirical Rademacher complexity is concerned is the behavior of a function class on those data points. We can define the empirical Rademacher complexity in terms of the shattering coefficient.

\begin{lemma} \label{lemma:rad_shat_coef}
	Let $\mathcal{Y} \subset \mathbb{R}$ be a finite set of real numbers of modulus at most $C > 0$. Given a sample $W = (W_{1},\,\dots,\,W_{n})$, the Rademacher complexity of any function class $\mathcal{H} \subseteq \{h: W \mapsto \mathcal{Y} \}$ can be bounded in terms of its shattering coefficient $s(\mathcal{H},\,n)$ by
	\begin{equation*}
		\hat{R}(\mathcal{H}) \leq C\sqrt{\frac{2\log s(\mathcal{H},\,n)}{n}}.
	\end{equation*}
\end{lemma}

\begin{proof}
	Let $A = \{h(W_{1}),\,\dots,\,h(W_{n})\,|\,h\in\mathcal{H}\}$, then $M = \sup_{a \in A} ||a|| = C\sqrt{n}$ and $|A| = s(\mathcal{H},\,n)$. Applying the Massart's lemma gives
	\begin{equation*}
		\hat{R}(\mathcal{H}) = \mathbb{E}_{\sigma}\left[\sup_{h \in \mathcal{H}} \frac{1}{n}\sum_{i=1}^{n}\sigma_{i}h(W_{i}) \right] \leq \frac{M\sqrt{2\log|\mathcal{H}|}}{n} = C\sqrt{\frac{2\log s(\mathcal{H},\,n)}{n}}.
	\end{equation*}
\end{proof}
Note that we apply the Massart's lemma conditional on the sample, hence, we can use the same bound for $\hat{R}(\mathcal{H})$. We can loosen the bound by applying Sauer's lemma which says that $s(\mathcal{H},\,n) \leq n^{d}$, where $d$ is the VC dimension of $\mathcal{H}$. This simplifies the result of Theorem \ref{lemma:rad_shat_coef} to
\begin{equation} \label{eq:rad_vc_bound}
	\hat{R}(\mathcal{H}) \leq C\sqrt{\frac{2d\log(n)}{n}} = O\left(\sqrt{\frac{\log(n)}{n}}\right).
\end{equation}

The bound in \eqref{eq:rad_vc_bound} is valid for any class with finite VC dimension which is coherent with the results of \cite{zhang_yu2005boostingES}. However, the VC bound is slower that the bound in \eqref{eq:rad_bound} by the factor or $\log(n)$ which appears in a lot of ML algorithms. 

Note that the bound in \eqref{eq:rad_vc_bound} is still a valid bound for the main results to follow. It only affects the rate of convergence.

\end{appendices}

\printbibliography

\end{document}